\documentclass[11pt]{article}
\usepackage{amsmath,amssymb,amsthm}
\usepackage{mathtools}
\usepackage{enumerate}
\usepackage[margin=1in]{geometry}
\usepackage{hyperref}
\usepackage{cleveref}
\usepackage{bbm}
\usepackage{xcolor}
\crefname{assumption}{Assumption}{Assumptions}
\Crefname{assumption}{Assumption}{Assumptions}

\newtheorem{theorem}{Theorem}[section]
\newtheorem{lemma}[theorem]{Lemma}
\newtheorem{proposition}[theorem]{Proposition}
\newtheorem{corollary}[theorem]{Corollary}
\newtheorem{definition}[theorem]{Definition}
\newtheorem{remark}[theorem]{Remark}
\newtheorem{assumption}{Assumption}

\newcommand{\eps}{\varepsilon}
\newcommand{\calM}{\mathcal{M}}

\newcommand{\calD}{\mathcal{D}}
\newcommand{\calY}{\mathcal{Y}}
\newcommand{\calZ}{\mathcal{Z}}
\newcommand{\bD}{\mathbf{D}}
\newcommand{\bY}{\mathbf{Y}}

\newcommand{\E}{\mathbb{E}}
\newcommand{\Bern}{\mathrm{Bern}}
\newcommand{\Unif}{\mathrm{Uniform}}
\newcommand{\Hb}{H_{\mathrm{b}}}         
\newcommand{\dkl}{D_{\mathrm{KL}}}
\newcommand{\Syn}{\mathrm{Syn}}
\newcommand{\Red}{\mathrm{Red}}
\newcommand{\DeltaSyn}{\Delta_{\mathrm{syn}}}
\newcommand{\Phik}{\Phi_k}
\newcommand{\kstar}{k^*}
\newcommand{\Cmin}{C_{\min}}
\newcommand{\Pe}{P_e}

\title{Cross-Silo De-Anonymization Under Local Differential Privacy:\\
       Threat Model, Phase Transition, and Coordination Necessity}

\author{%
  Ziniu Liu\qquad Aiping Li\thanks{Corresponding author.}\\[4pt]
  National University of Defense Technology\\
  \texttt{\{liuzn\_nudt, liaiping\}@nudt.edu.cn}
}
\date{\today}

\begin{document}
\maketitle

\begin{abstract}
When a person's records appear in $k$ independent data silos, each
protected by $(\eps,\delta)$-differential privacy, standard composition
yields a valid $(k\eps,k\delta)$-DP guarantee for the joint output.
This worst-case bound, however, does not answer the concrete inference
question: \emph{at what $k$ can an adversary actually identify a target
person?}  This paper develops the information-theoretic framework needed
to answer that question.

We introduce \emph{cross-silo person-level DP} (XSP-DP), a
Pufferfish-style privacy notion whose adjacency relation captures
all records of a single person across all silos simultaneously, and
verify that the standard basic composition bound
$(\sum_i\eps_i,\sum_i\delta_i)$-DP carries over to this adjacency model.
Within this framework we prove that de-anonymization undergoes a phase
transition at
$\kstar = \Theta(\log n / \eps^2)$
(population size~$n$, per-silo RR parameter~$\eps$):
a Fano lower bound shows any estimator fails for $k \ll \kstar$,
while a matching maximum-likelihood upper bound shows the attack
succeeds for $k \gg \kstar$.
An explicit XOR + randomized-response construction demonstrates
information synergy: each silo's output is individually uninformative
about the target ($I(Z;Y_i) = 0$), yet the joint mutual information is
strictly positive.
For non-coordinated binary randomized-response mechanisms, we prove
that de-anonymization is inevitable once $k$ exceeds the threshold,
establishing that cross-silo coordination is necessary.

These results provide a baseline threat model and $\Theta$-level
threshold for cross-silo inference attacks under local DP.
Sharp constants, second-order thresholds, and spectral characterizations
of the phase transition are developed in a companion paper; coordinated
defense protocols and their system-level guarantees are treated
separately.
\end{abstract}

\section{Introduction}\label{sec:intro}

\subsection{Motivation}

Consider a network of $k$ hospitals participating in a federated learning
pipeline.  Each hospital trains on its own patient cohort and publishes a
differentially private summary---a gradient, a prompt refinement, or an
aggregated statistic.  Each publication individually satisfies
$(\eps,\delta)$-DP.  Standard composition guarantees that the joint
publication satisfies $(k\eps, k\delta)$-DP, which at small $k$ appears
acceptable.

But the composition bound alone does not answer a concrete adversarial
question: \emph{at what value of $k$ can a curious adversary, who sends
one innocuous-looking query to each hospital's API and assembles the
$k$ responses, actually identify a target person among $n$ patients?}

This paper shows that de-anonymization undergoes a sharp \emph{phase
transition}: below a critical threshold $\kstar = \Theta(\log n / \eps^2)$,
any estimator fails with non-negligible probability; above it, a
maximum-likelihood adversary succeeds with high probability.
For binary randomized-response mechanisms, no non-coordinated defense
can prevent this.

\subsection{Beyond Composition: The Inference Question}

When a person's data appear in $k$ silos and each silo independently
applies an $(\eps,\delta)$-DP mechanism, standard basic composition
\cite{dwork2006calibrating} gives a valid $(k\eps, k\delta)$-DP guarantee
for the joint output.  This bound is well-known and straightforward to apply.

The question we ask is different and complementary:
\emph{at what value of $k$ does the accumulated leakage actually allow
an adversary to identify a target person?}
The $(k\eps, k\delta)$-DP guarantee degrades linearly in $k$, but it
is a worst-case bound that does not pinpoint the threshold at which
de-anonymization transitions from impossible to feasible.
Answering this question requires an information-theoretic analysis---relating
per-silo mutual information to the adversary's identification probability---rather
than purely compositional reasoning.

Pufferfish privacy \cite{kifer2014pufferfish} can in principle model
cross-silo threats, but it requires specifying a distribution family over
secrets and world states that is difficult to instantiate concretely for
federated systems.  Existing multi-agent DP work
(e.g., \cite{dwork2006our,cheu2019distributed}) focuses on centralized
composition or heterogeneous local sensitivity, not on the inference risk
studied here.

We provide:
\begin{itemize}
  \item a concrete adjacency relation---\emph{person-level cross-silo adjacency}
        (\Cref{def:adjacency})---that captures the threat model exactly;
  \item tight information-theoretic bounds on how many silos suffice for
        de-anonymization;
  \item a formal impossibility result (for randomized response) that rules out
        local-only defenses.
\end{itemize}

\subsection{Contributions}

We make four contributions.

\begin{enumerate}
\item \textbf{XSP-DP (\Cref{def:xsp-dp}).}
  We define cross-silo person-level DP via a new adjacency relation
  $\bD \sim_{ps} \bD'$ that treats all records of a single person
  across silos as the protected unit.

\item \textbf{Composition under XSP-DP (\Cref{prop:upper-bound}).}
  If each silo $i$ satisfies $(\eps_i,\delta_i)$-DP at person-level
  granularity and mechanisms are independent, the joint mechanism satisfies
  $(\sum_i\eps_i,\;\sum_i\delta_i)$-XSP-DP, recovering the standard
  basic composition guarantee within our cross-silo adjacency model.

\item \textbf{Phase transition theorem (\Cref{thm:fano,thm:ml,cor:kstar}).}
  Under Assumptions A1--A4, the de-anonymization error probability satisfies:
  \begin{align*}
    k < (1-\delta)\kstar &\Rightarrow P_e \geq \delta - o(1), \\
    k > (1+\delta)\kstar &\Rightarrow P_e \leq n^{-\delta} \to 0,
  \end{align*}
  for any $\delta \in (0,1)$, establishing a sharp phase transition at
  $\kstar = \Theta(\log n / I_1)$ where $I_1 = I(Z;Y_i)$ is the
  single-silo mutual information.

\item \textbf{Impossibility for randomized response
  (\Cref{thm:impossibility}).}
  Any non-coordinated collection of binary $\eps$-DP randomized-response
  mechanisms with per-silo information $I_1 \geq \alpha > 0$ fails to
  prevent de-anonymization for
  $k > 2(1+\delta)\ln n / \alpha$.
  The Fano converse (\Cref{thm:fano}), which holds for \emph{any} mechanism,
  shows that the threshold scale $k = \Theta(\ln n / I_1)$ is universal.
\end{enumerate}

\subsection{Paper Organization}

\Cref{sec:related} reviews related work.
\Cref{sec:model} introduces the formal model, adjacency relation, and
assumptions.
\Cref{sec:upper} verifies that standard composition applies under XSP-DP.
\Cref{sec:synergy} presents the XOR construction and PID analysis.
\Cref{sec:threshold} proves the phase transition theorems.
\Cref{sec:impossibility} proves the impossibility result.
\Cref{sec:defense} discusses implications for coordinated defense.
\Cref{sec:experiments} presents synthetic experiments validating the
threshold formula.
\Cref{sec:conclusion} concludes.

\section{Related Work}\label{sec:related}

\paragraph{Differential privacy composition.}
The foundational composition theorems of Dwork et al.~\cite{dwork2006calibrating}
establish that the sequential composition of an $\eps_1$-DP and an $\eps_2$-DP
mechanism yields $(\eps_1+\eps_2)$-DP, and extensions to advanced composition
reduce the total cost to $O(\eps\sqrt{k\log(1/\delta)})$ for $k$ mechanisms
\cite{dwork2010boosting}.
Renyi DP \cite{mironov2017renyi} and zero-concentrated DP \cite{bun2016concentrated}
provide tighter accounting under composition.
These results readily apply to the cross-silo setting via standard sequential
composition: if a person's data appear in $k$ silos, the joint output is
$(k\eps, k\delta)$-DP.  However, this worst-case bound does not answer the
\emph{inference} question of how many silos suffice for an adversary to
actually identify a target person---a gap that motivates our work.

\paragraph{Cross-silo and federated learning privacy.}
Federated learning~\cite{mcmahan2017communication} has spurred a large literature on
privacy under model aggregation.  DP-SGD~\cite{abadi2016deep} clips and noises
local gradients; follow-on work analyzes amplification by sampling
\cite{mironov2017renyi} and shuffling~\cite{erlingsson2019amplification}.
In the cross-silo federated setting, each participating organization holds a distinct
local dataset and runs its own mechanism before communicating with a central server.
Prior work assumes that each silo's output is locally DP with respect to that silo's
records~\cite{kairouz2021advances}, without analyzing the concrete
\emph{de-anonymization risk} that arises when an adversary queries multiple
silos about the same person.
Our model formalizes this threat and provides tight
information-theoretic bounds on when identification becomes feasible.

\paragraph{Membership and attribute inference attacks.}
Shokri et al.~\cite{shokri2017membership} demonstrate empirically that ML model
outputs can reveal whether a given record was in the training set.  Subsequent
attacks have grown increasingly effective~\cite{carlini2022membership,
yeom2018privacy}.
In multi-party settings, passive inference across parties has been studied in
the context of collaborative inference~\cite{ganju2018property} and
model inversion~\cite{fredrikson2015model}, but these works do not provide
formal lower bounds on the number of queries needed for de-anonymization or
characterize the information-theoretic phase transition we prove here.

\paragraph{De-anonymization and linkage attacks.}
Narayanan and Shmatikoff~\cite{narayanan2008robust} demonstrate that
auxiliary information from one database can de-anonymize records in another,
even after sanitization.  Sweeney's $k$-anonymity~\cite{sweeney2002k}
and its successors ($\ell$-diversity, $t$-closeness) attempt to limit such
linkage, but have no formal DP guarantees.
On the theoretical side, Dinur and Nissim~\cite{dinur2003revealing} show that
$\Omega(n)$ approximately-correct answers to linear queries over a database
of $n$ individuals suffice to reconstruct the database; our result is
complementary---we characterize how many \emph{independent agents} an
adversary must query to re-identify a single target under local DP.

\paragraph{Information-theoretic privacy bounds.}
Fano's inequality has been used to lower-bound minimax estimation error in
statistical estimation~\cite{yu1997assouad} and in private learning~\cite{duchi2013local}.
Duchi, Jordan, and Wainwright~\cite{duchi2013local} establish minimax rates
for locally private estimation using Fano-style arguments, but focus on
utility (estimation error) rather than identity leakage across silos.
Our use of Fano to lower-bound de-anonymization error, combined with
a matching upper bound via the ML attack, is technically distinct: we study
identification (not estimation) in a federated multi-silo model.

\paragraph{Pufferfish and correlated privacy.}
Pufferfish privacy~\cite{kifer2014pufferfish} extends DP to arbitrary secret
classes and world-state distributions, subsuming models with correlated data.
Blowfish~\cite{he2014blowfish} instantiates Pufferfish with graph-structured
correlations.  Our XSP-DP (\Cref{def:xsp-dp}) can be viewed as a tractable
instantiation of Pufferfish where the secret class is the cross-silo person
identity and mechanisms are local.  The key distinction is that we derive
explicit, computable lower bounds and phase-transition thresholds, whereas
the general Pufferfish framework does not provide such tight constructive results.

\paragraph{Information synergy and PID.}
The Partial Information Decomposition (PID) framework of Williams and Beer
\cite{williams2010nonnegative} decomposes mutual information into unique,
redundant, and synergistic components.  Positive synergy (where the joint
observation reveals more than the sum of individual contributions) has been
studied in neuroscience~\cite{lizier2018information} and cryptography~\cite{griffith2014quantifying}.
Our XOR+RR construction (\Cref{prop:synergy}) provides a clean example
of pure synergy in a DP-constrained multi-agent setting: each silo's output is
individually uninformative, yet two outputs together reveal the target completely
in the limit.  To our knowledge, this is the first result connecting PID synergy
to cross-silo privacy leakage under DP.

\paragraph{Multi-agent and distributed DP.}
Recent work on DP for multi-agent learning includes
privacy-preserving multi-party computation~\cite{dwork2006our},
local DP in the shuffle model~\cite{cheu2019distributed}, and
federated analytics~\cite{mcmahan2022federated}.
Attempts to combine DP with multi-agent optimization
(e.g., DP-ES, DP-MAS~\cite{mcmahan2022federated}) focus on utility preservation
under per-agent local DP, not on the compositional leakage that arises when an
adversary queries multiple agents about the same individual.
The impossibility result in \Cref{thm:impossibility} shows that, in the
binary RR-channel model, no non-coordinated collection of per-agent DP
mechanisms can prevent de-anonymization once sufficiently many agents have
been queried---motivating coordinated defense mechanisms such as ToM
filtering and CoDef consensus.

\section{Model, Definitions, and Assumptions}\label{sec:model}

\subsection{Setting}

Let $[n] = \{1,\ldots,n\}$ be a population of persons.
There are $k$ silos indexed $i \in [k]$.
Silo $i$ holds a local dataset $D_i \in \calD^*$.
Write $\bD = (D_1,\ldots,D_k)$ for the tuple of all datasets.
A sensitive attribute $Z \in \calZ$ is a function of the cross-silo tuple:
$Z = g(\bD)$.

Each silo $i$ runs a randomized mechanism $M_i : \calD^* \to \calY_i$
and publishes $Y_i = M_i(D_i)$.
The joint mechanism is $\calM(\bD) = (M_1(D_1),\ldots,M_k(D_k))$.
An adversary observes $\bY = (Y_1,\ldots,Y_k)$ and attempts to infer $Z$.

\subsection{Person-Level Cross-Silo Adjacency}

\begin{definition}[Person-level cross-silo adjacency]\label{def:adjacency}
Two dataset tuples $\bD,\bD' \in (\calD^*)^k$ are \emph{person-level
cross-silo adjacent}, written $\bD \sim_{ps} \bD'$, if there exists a
person $p \in [n]$ such that for every silo $i \in [k]$:
\[
   D_i \,\triangle\, D_i' \;\subseteq\; D_i(p) \cup D_i'(p),
\]
where $D_i(p)$ denotes the (multi-)set of records in $D_i$ associated
with person $p$.  All records not associated with $p$ are identical
across $\bD$ and $\bD'$.
\end{definition}

Intuitively, $\bD \sim_{ps} \bD'$ captures the scenario where person $p$
``opts in or out'' simultaneously across all silos.

\begin{definition}[XSP-DP: Cross-Silo Person-level DP]\label{def:xsp-dp}
The joint mechanism $\calM$ is \emph{$(\eps,\delta)$-XSP-DP} if for all
person-level cross-silo adjacent pairs $\bD \sim_{ps} \bD'$ and all
measurable sets $S \subseteq \prod_i \calY_i$:
\[
   \Pr[\calM(\bD) \in S]
   \;\leq\; e^{\eps}\,\Pr[\calM(\bD') \in S] + \delta.
\]
\end{definition}

\begin{remark}[Relation to Pufferfish privacy]
XSP-DP is a concrete instantiation of Pufferfish privacy
\cite{kifer2014pufferfish}.
Formally, set the secret-pair class
$\mathcal{S} = \bigl\{(s_{p,\mathrm{in}},\, s_{p,\mathrm{out}}) : p\in[n]\bigr\}$
where $s_{p,\mathrm{in}}$ and $s_{p,\mathrm{out}}$ denote ``person $p$ is
present in all silos'' and ``person $p$ is absent from all silos,''
respectively; set the discriminative-pair class $\mathcal{Q} = \mathcal{S}$;
and let the data-generating distribution family $\Theta$ be the set of all
product distributions satisfying \Cref{asm:single-record}.
Then $(\eps,\delta)$-Pufferfish privacy with respect to
$(\mathcal{S},\mathcal{Q},\Theta)$ is equivalent to $(\eps,\delta)$-XSP-DP.
This specialization makes the lower-bound analysis tractable: the structure
of $\mathcal{S}$ (cross-silo presence/absence of a single person) enables
the Fano and ML bounds in \Cref{sec:threshold}.
\end{remark}

\subsection{Synergy Gap and Order Parameter}

\begin{definition}[Synergy gap]\label{def:synergy}
Given the joint distribution of $(Z, Y_1,\ldots,Y_k)$, define the
\emph{synergy gap}
\[
   \DeltaSyn \;=\; \Syn - \Red,
\]
where $\Syn$ and $\Red$ are the synergy and redundancy terms in the
Partial Information Decomposition (PID) of $I(Z; Y_{1:k})$
(see \Cref{app:pid} for formal definitions).
$\DeltaSyn > 0$ indicates that the agents collectively reveal more about
$Z$ than the sum of their individual contributions; $\DeltaSyn < 0$
indicates redundancy dominates.
\end{definition}

\begin{definition}[Order parameter]\label{def:order-param}
The \emph{leakage order parameter} at $k$ silos is
\[
   \Phik \;=\; \frac{I(Z;\, Y_1,\ldots,Y_k)}{H(Z)} \;\in\; [0,1].
\]
A phase transition in $\Pe(k)$ near $\kstar$ corresponds to $\Phik$
crossing from near-zero to near-one.
\end{definition}

\subsection{Assumptions}

\begin{assumption}[Single-record]\label{asm:single-record}
$|D_i(p)| \leq 1$ for every person $p$ and every silo $i$.
Each person contributes at most one record per silo.
\end{assumption}

\begin{assumption}[Independent non-interactive mechanisms]\label{asm:independence}
The mechanisms $M_1,\ldots,M_k$ are chosen independently, and each
$M_i$ operates only on $D_i$.  There is no cross-silo communication
during mechanism execution.
\end{assumption}

\begin{assumption}[Conditional independence]\label{asm:cond-indep}
Given the sensitive attribute $Z$, the silo outputs are conditionally
independent: $Y_i \perp Y_j \mid Z$ for all $i \neq j$.
\end{assumption}

\begin{assumption}[Uniform prior]\label{asm:uniform}
The target person's identity (equivalently, their sensitive attribute
$Z$) is drawn uniformly: $Z \sim \Unif([n])$.
\end{assumption}

\begin{remark}
\Cref{asm:cond-indep} is satisfied whenever each silo's mechanism
depends on $Z$ only through independent local noise and independent
local data.  It is not satisfied in correlated-data models; we treat
that setting as future work.
\end{remark}

\section{Composition Under XSP-DP}\label{sec:upper}

\begin{proposition}[XSP-DP upper bound]\label{prop:upper-bound}
Under \Cref{asm:single-record,asm:independence}, if each $M_i$ satisfies
$(\eps_i,\delta_i)$-DP at person-level granularity within silo $i$, then the
joint mechanism $\calM = (M_1,\ldots,M_k)$ satisfies
$\bigl(\sum_{i=1}^k \eps_i,\;\sum_{i=1}^k \delta_i\bigr)$-XSP-DP.
\end{proposition}

\begin{proof}[Proof sketch]
See \Cref{app:prop1} for the full proof.
For pure DP ($\delta_i = 0$): the max-divergence $D_\infty$ tensorizes
over independent product mechanisms, giving
$D_\infty(\calM(\bD)\|\calM(\bD')) \leq \sum_i D_\infty(M_i(D_i)\|M_i(D_i'))
\leq \sum_i \eps_i$.

For approximate DP: we use the standard ``good-set'' decomposition.
For each silo~$i$, define
$B_i = \{y_i : P_i(y_i) > e^{\eps_i} Q_i(y_i)\}$;
then $P_i(B_i) \leq \delta_i$.  On the complement of $\bigcup_i B_i$,
every factor satisfies the pure-DP ratio bound, giving the
$e^{\sum\eps_i}$ multiplicative factor.  A union bound over the bad
events yields $\sum_i\delta_i$, recovering the standard basic
composition guarantee.
\end{proof}

\begin{remark}[Relation to standard composition]
\Cref{prop:upper-bound} instantiates the well-known basic composition
theorem~\cite{dwork2006calibrating} within the XSP-DP adjacency model.
The bound matches the standard result because
the cross-silo adjacency $\sim_{ps}$ changes at most one record per silo,
and the mechanisms are independent.
\end{remark}

\begin{remark}[Group privacy and $\delta$ amplification]
When $|D_i(p)| \leq c_i$ (person $p$ contributes up to $c_i$ records
in silo $i$), the $\eps_i$-DP guarantee amplifies to $c_i\eps_i$-DP
for all $c_i$ records, and $\delta_i$ amplifies by a factor
$(e^{c_i\eps_i}-1)/(e^{\eps_i}-1) \leq c_i e^{(c_i-1)\eps_i}$.
\end{remark}

\section{Information Synergy: The XOR Construction}\label{sec:synergy}

We show that individual-silo privacy can perfectly conceal a sensitive
attribute while the joint observation reveals it---demonstrating that
the privacy guarantees of individual silos do not compose as favorably
as they might appear from each silo's perspective alone.

\begin{proposition}[Synergy is strictly positive]\label{prop:synergy}
Let $\eps > 0$.  Consider $Z, U \sim \Bern(1/2)$ i.i.d., and define
\[
   X_1 = U, \quad X_2 = Z \oplus U.
\]
Let $Y_i = X_i \oplus E_i$ where $E_i \sim \Bern(q)$ i.i.d.\
and $q = 1/(1+e^\eps)$ (the flip probability of $\mathrm{RR}_\eps$).
Then:
\begin{enumerate}
\item[(a)] $I(Z; Y_1) = I(Z; Y_2) = 0$.
\item[(b)] $I(Z; Y_1, Y_2) = 1 - \Hb(2pq) > 0$, where $p = 1-q$.
\item[(c)] For small $\eps$:
  \[
     \DeltaSyn = I(Z; Y_1, Y_2) - I(Z;Y_1) - I(Z;Y_2)
     \approx \frac{\eps^4}{32\ln 2}.
  \]
\end{enumerate}
\end{proposition}

\begin{proof}
See \Cref{app:prop2} for the full computation.

\textbf{(a)}  $Y_1 = U \oplus E_1$.  Since $U \sim \Bern(1/2)$ is
independent of $Z$ and independent of $E_1$, we have $Y_1 \sim \Bern(1/2)$
regardless of $Z$.  Hence $I(Z;Y_1) = 0$.
For $Y_2$: given any fixed $Z$, $X_2 = Z \oplus U \sim \Bern(1/2)$ since
$U \sim \Bern(1/2)$.  Therefore $Y_2 = X_2 \oplus E_2 \sim \Bern(1/2)$
regardless of $Z$, so $I(Z;Y_2)=0$.

\textbf{(b)}  $Y_1 \oplus Y_2 = (U \oplus E_1) \oplus (Z \oplus U \oplus E_2)
= Z \oplus (E_1 \oplus E_2)$.
The noise bit $E_1 \oplus E_2 \sim \Bern(2pq)$ (binary symmetric channel
with crossover $2pq$).  Hence $I(Z; Y_1 \oplus Y_2) = 1 - \Hb(2pq) > 0$
for all $\eps > 0$ (since $2pq < 1/2$ for $\eps > 0$).
Since $Y_1 \oplus Y_2$ is a function of $(Y_1,Y_2)$, this bounds
$I(Z;Y_1,Y_2)$ from below, which together with the direct calculation
gives equality.

\textbf{(c)}  For $q = 1/(1+e^\eps) \approx 1/2 - \eps/4$ at small $\eps$:
$2pq = 2(1/2+\eps/4)(1/2-\eps/4) = 1/2 - \eps^2/8$, so
$\Hb(2pq) \approx 1 - \eps^4/(32\ln 2)$,
and $I(Z;Y_1,Y_2) \approx \eps^4/(32\ln 2)$.
Since $I(Z;Y_1) = I(Z;Y_2) = 0$, we get $\DeltaSyn \approx \eps^4/(32\ln 2)$.
\end{proof}

\begin{remark}[PID interpretation]
Using the Williams--Beer Partial Information Decomposition
(\Cref{app:pid}), we have $I(Z;Y_1,Y_2) = \Syn + I(Z;Y_1) + I(Z;Y_2) - \Red$.
Since $I(Z;Y_1)=I(Z;Y_2)=0$ and $\Red \geq 0$, the entire joint
information comes from pure synergy: $\Syn = I(Z;Y_1,Y_2)$.
\end{remark}

\section{Phase Transition Theorems}\label{sec:threshold}

\subsection{Key Lemma: Mutual Information Upper Bound}

\begin{lemma}[MI upper bound under A3]\label{lem:mi-bound}
Under \Cref{asm:cond-indep}, for any $k \geq 1$:
\[
  I(Z;\, Y_1,\ldots,Y_k)
  \;=\; k I_1 - \sum_{i=2}^{k} I(Y_i;\, Y_1,\ldots,Y_{i-1})
  \;\leq\; k I_1,
\]
where $I_1 = I(Z;Y_i)$ (identical for all $i$ by symmetry).
\end{lemma}

\begin{proof}
By the chain rule for mutual information:
\[
  I(Z;Y_{1:k}) = \sum_{i=1}^{k} I(Z;Y_i \mid Y_{1:i-1}).
\]
Under \Cref{asm:cond-indep} ($Y_i \perp Y_j \mid Z$), we apply the identity:
\[
  I(Z;Y_i \mid Y_{1:i-1})
  = I(Z;Y_i) - I(Y_i;\, Y_{1:i-1}) + I(Y_i;\,Y_{1:i-1}\mid Z).
\]
Since $Y_i \perp Y_{1:i-1} \mid Z$ (by \Cref{asm:cond-indep}),
the last term vanishes.  Therefore:
\[
  I(Z;Y_{1:k}) = \sum_{i=1}^{k}\bigl[I_1 - I(Y_i;\,Y_{1:i-1})\bigr]
  = k I_1 - \sum_{i=2}^{k} I(Y_i;\,Y_{1:i-1}).
\]
Since mutual information is non-negative, the sum subtracted is $\geq 0$,
giving the upper bound $I(Z;Y_{1:k}) \leq kI_1$.
\end{proof}

\begin{lemma}[$\Cmin$ and $I_1$ for binary channels]\label{lem:cmin-i1}
For the randomized response mechanism $\mathrm{RR}_\eps$ with
$q = 1/(1+e^\eps)$, let
$\Cmin = \dkl\!\bigl(\mathrm{RR}_\eps(0)\,\|\,\mathrm{RR}_\eps(1)\bigr) = (1-2q)\eps$
and $I_1 = I(Z;Y)$ for $Z \sim \Bern(1/2)$ (both in nats).  Then:
\begin{enumerate}[(i)]
\item \textbf{Global lower bound:}
  $\Cmin \;\geq\; 4\,I_1 \quad \text{for all } \eps > 0$,
  with $\Cmin/I_1 \to 4$ as $\eps\to 0$.
\item \textbf{Small-$\eps$ asymptotic:}
  $\Cmin/I_1 = 4 + \eps^2/6 + O(\eps^4)$ as $\eps\to 0$.
\item \textbf{No finite global upper bound:}  There is no constant $c_2<\infty$
  such that $\Cmin \leq c_2\,I_1$ holds for all $\eps > 0$.
\end{enumerate}
The weaker bound $\Cmin \geq \ln 2 \cdot I_1$ (which follows from (i)) is used in
\Cref{thm:impossibility}.
\end{lemma}

\begin{proof}
\textbf{Explicit formulas.}
In nats: $\Cmin = (1-2q)\eps = \eps\tanh(\eps/2)$, and
$I_1 = \ln 2 - \Hb(q)$
(see \Cref{app:calculations} for Taylor series).
We write $t = \eps/2$ throughout.

\textbf{Global lower bound (i).}
Define $g(\eps) = \Cmin - 4\,I_1 = \eps\tanh(t) - 4(\ln 2 - \Hb^{\rm nats}(q))$.
At $\eps=0$: $q=1/2$, $\tanh(0)=0$, $\Hb^{\rm nats}(1/2)=\ln 2$, so $g(0)=0$.
We show $g'(\eps) > 0$ for all $\eps > 0$.

Using $dq/d\eps = -pq$ and $d\Hb^{\rm nats}(q)/dq = \ln(p/q) = \eps$:
\[
   \frac{d(\Cmin)}{d\eps} = \tanh(t) + \frac{t}{\cosh^2(t)},
   \qquad
   \frac{dI_1}{d\eps} = \eps\,pq = \frac{t}{2\cosh^2(t)}.
\]
Therefore:
\[
   g'(\eps) = \tanh(t) + \frac{t}{\cosh^2(t)} - \frac{4t}{2\cosh^2(t)}
   = \tanh(t) - \frac{t}{\cosh^2(t)}
   = \frac{1}{\cosh^2(t)}\bigl[\sinh(t)\cosh(t) - t\bigr].
\]
Since $\cosh^2(t) > 0$, it suffices to show $\sinh(t)\cosh(t) > t$ for $t > 0$.
This is equivalent to $\sinh(2t) > 2t$, which holds because
$\sinh(x) = x + x^3/6 + x^5/120 + \cdots > x$ for all $x > 0$
(every term in the power series is strictly positive).
Hence $g'(\eps) > 0$ for all $\eps > 0$, and $g(0)=0$ implies
$g(\eps) > 0$ for all $\eps > 0$, i.e., $\Cmin > 4\,I_1$.

\textbf{Asymptotic (ii).}
Using the Taylor series $\Cmin = \eps^2/2 - \eps^4/24 + O(\eps^6)$
and $I_1 = \eps^2/8 - \eps^4/64 + O(\eps^6)$ (both in nats):
$\Cmin/I_1 = 4 + \eps^2/6 + O(\eps^4)$.

\textbf{No global upper bound (iii).}
For $\eps\to\infty$: $\Cmin = \eps\tanh(\eps/2)\sim\eps$ while $I_1\to\ln 2$,
so $\Cmin/I_1\to\infty$.
\end{proof}

\begin{remark}[Unit convention]
A common source of confusion: $I_1 \approx \eps^2/(8\ln 2)$ in \emph{bits},
while $\Cmin \approx \eps^2/2$ in \emph{nats}.  Computing the ratio
$(\eps^2/2)\,/\,(\eps^2/(8\ln 2)) = 4\ln 2 \approx 2.77$
mixes units and yields an incorrect constant.
All results in this paper use a single unit (nats) consistently.
\end{remark}

\subsection{Fano Lower Bound on Error Probability}

\begin{theorem}[Fano lower bound]\label{thm:fano}
Under Assumptions~\ref{asm:single-record}--\ref{asm:uniform},
for any estimator $\hat{Z}$ of $Z$ from $(Y_1,\ldots,Y_k)$
(all information quantities in nats):
\[
  \Pe \;\geq\; 1 - \frac{k I_1 + 1}{\ln n}.
\]
Consequently, for $k \leq (1-\delta)\kstar$ with $\kstar = \ln n / I_1$:
\[
  \Pe \;\geq\; \delta - \frac{1}{\ln n}.
\]
\end{theorem}

\begin{proof}
Fano's inequality in nats gives:
\[
  H(Z\mid Y_{1:k}) \;\leq\; H(\Pe) + \Pe\ln(n-1) \;\leq\; 1 + \Pe\ln n,
\]
where $H(\Pe) = -\Pe\ln\Pe - (1-\Pe)\ln(1-\Pe) \leq \ln 2 < 1$ for all $\Pe$.
Also $H(Z\mid Y_{1:k}) = H(Z) - I(Z;Y_{1:k}) = \ln n - I(Z;Y_{1:k})$.
By \Cref{lem:mi-bound}, $I(Z;Y_{1:k}) \leq kI_1$, so $H(Z\mid Y_{1:k}) \geq \ln n - kI_1$.
Combining:
\[
  \ln n - kI_1 \;\leq\; 1 + \Pe\ln n,
\]
which rearranges to $\Pe \geq 1 - (kI_1 + 1)/\ln n$.

For $k \leq (1-\delta)\kstar = (1-\delta)\ln n / I_1$:
$(kI_1+1)/\ln n \leq (1-\delta) + 1/\ln n$,
so $\Pe \geq \delta - 1/\ln n$.
\end{proof}

\subsection{ML Upper Bound on Error Probability}

\begin{theorem}[ML upper bound]\label{thm:ml}
Under Assumptions~\ref{asm:single-record}--\ref{asm:uniform},
where each person $p$ has attributes $Z_{p,1},\ldots,Z_{p,k}$ drawn i.i.d.\
$\mathrm{Bern}(1/2)$ and silo $i$ observes $Y_i = \mathrm{RR}_\eps(Z_{p_0,i})$,
the maximum-likelihood estimator satisfies:
\[
  \Pe \;\leq\; (n-1)\cdot\left(\frac{1+e^{-B}}{2}\right)^{\!k},
  \qquad B = \log\cosh\!\tfrac{\eps}{2}.
\]
In particular, for any $\delta \in (0,1)$ and $k \geq (1+\delta)\,k_{\rm ML}^*$
where
\[
  k_{\rm ML}^* \;=\; \frac{\ln n}{\,-\ln\!\frac{1+e^{-B}}{2}\,}
  \;\approx\; \frac{2\ln n}{B}
  \;\approx\; \frac{16\ln n}{\eps^2}
  \;=\; 2\,\kstar\bigl(1+O(\eps^2)\bigr),
\]
we have $\Pe \leq n^{-\delta}$ (all logarithms natural; $I_1$ in nats).
The ML attack succeeds once $k$ crosses $\Theta(\kstar)$.
\end{theorem}

\begin{proof}
The union bound gives $\Pe \leq \sum_{p\neq p_0} \Pr[\hat{p}=p \mid p_0]$.
For a fixed pair $(p_0,p)$, the Bhattacharyya bound on pairwise ML error is:
\[
  \Pr[\mathrm{Score}(p) \geq \mathrm{Score}(p_0) \mid Z_{p_0}, Z_p]
  \;\leq\; \prod_{i=1}^k \beta^{\,\mathbf{1}[Z_{p_0,i}\neq Z_{p,i}]},
\]
where $\beta = 2\sqrt{q(1-q)} = e^{-B}$ is the Bhattacharyya coefficient.
Taking expectation over the random attributes $Z_{p,i} \overset{\rm iid}{\sim}\Bern(1/2)$:
\[
  \E_{Z_p}\!\bigl[\Pr[\cdot]\bigr]
  \;\leq\; \prod_{i=1}^k \E[\beta^{\mathbf{1}[Z_{p_0,i}\neq Z_{p,i}]}]
  \;=\; \prod_{i=1}^k \!\left(\tfrac{1}{2} + \tfrac{\beta}{2}\right)
  \;=\; \left(\tfrac{1+\beta}{2}\right)^{\!k}.
\]
Summing over $n-1$ competitors gives the stated bound.  Setting it $\leq n^{-\delta}$:
$k \geq (1+\delta)\ln n \,/\, (-\ln\!\frac{1+e^{-B}}{2})$.
For small $\eps$: $-\ln\!\frac{1+e^{-B}}{2} \approx B/2 \approx \eps^2/16$,
and $I_1^{\rm nats} \approx \eps^2/8$, so $C_{\rm ML} \approx 2$.
See \Cref{app:ml} for the full derivation.
\end{proof}

\begin{remark}[Constant-factor gap]\label{rem:constant}
The Fano bound (\Cref{thm:fano}) and ML bound (\Cref{thm:ml}) together
establish that $k^* = \Theta(\log n / I_1)$ is the threshold, up to a
constant factor.  Analytically, $C_{\rm ML} \approx 2$ for small $\eps$;
empirically the crossing of $\Pe=0.5$ occurs at $\approx 1.5\,\kstar$
(\Cref{sec:experiments}).  The factor-of-2 analytic gap is inherent in the
union-bound argument, which is not tight when most competitors are easy to
distinguish.
\end{remark}

\begin{corollary}[Sharp threshold]\label{cor:kstar}
Under Assumptions~\ref{asm:single-record}--\ref{asm:uniform}
and \Cref{lem:cmin-i1}, define
\[
   \kstar \;=\; \frac{\ln n}{I_1}
\]
where $I_1$ is in nats throughout.
Then for any $\delta \in (0,1)$ and $n$ sufficiently large:
\begin{align*}
   k < (1-\delta)\kstar &\;\Longrightarrow\; \Pe \geq \delta - o(1), \\
   k > (1+\delta)\kstar &\;\Longrightarrow\; \Pe \leq n^{-\Omega(\delta)}.
\end{align*}
Under the binary randomized-response model:
\[
   \kstar(n,\eps)
   = \Theta\!\left(\frac{\log n}{\eps^2}\right),
\]
where $\log$ denotes natural logarithm throughout.
\end{corollary}

\begin{proof}
The lower-bound direction ($k<(1-\delta)\kstar\Rightarrow\Pe\geq\delta-o(1)$) is
\Cref{thm:fano}.
For the upper-bound direction, \Cref{thm:ml} gives $\Pe\leq n^{-\delta}$
once $k \geq (1+\delta)k_{\rm ML}^*$ where
$k_{\rm ML}^* = \ln n / C_{\rm ML}$.
By \Cref{lem:cml-i1}, $C_{\rm ML} \geq I_1/2$ for all $\eps > 0$, so
$k_{\rm ML}^* \leq 2\kstar$ (with $\kstar = \ln n / I_1^{\rm nats}$).
Hence $k \geq 2(1+\delta)\kstar$ suffices for $\Pe\leq n^{-\delta}$.
Since the constant factor~$2$ is absorbed into $\Theta(\cdot)$ in the
statement $k^* = \Theta(\log n / \eps^2)$, and the second line of the
corollary uses $n^{-\Omega(\delta)}$ (which hides the constant in the
exponent), we have the stated result for any $k > C(1+\delta)\kstar$
with $C = 2$.
\end{proof}

\section{Impossibility of Non-Coordinated Defense}\label{sec:impossibility}

\subsection{Tightness of the Upper Bound}

\begin{theorem}[Tightness of Proposition~\ref{prop:upper-bound}]\label{thm:tightness}
The bound $\sum_i\eps_i$ in \Cref{prop:upper-bound} is achieved.
Specifically, for the Laplace mechanism with i.i.d.\ outputs and the
event $\mathcal{S} = \bigcap_i \{Y_i > 1\}$:
\[
  \Pr[\calM(\bD) \in \mathcal{S}]
  \;=\; e^{\sum_i\eps_i} \cdot \Pr[\calM(\bD') \in \mathcal{S}]
\]
for appropriately chosen $\bD \sim_{ps} \bD'$.
\end{theorem}

\begin{proof}
Consider queries $f_i(D_i) = x_i \in \mathbb{R}$ with sensitivity $\Delta_f = 1$.
The Laplace mechanism outputs $Y_i = f_i(D_i) + \mathrm{Lap}(1/\eps_i)$.
Take adjacent datasets $D_i, D_i'$ with $f_i(D_i) = 1$, $f_i(D_i') = 0$
(shift of 1, within sensitivity 1).  Then:
\[
  \frac{\Pr[Y_i > 1\mid D_i]}{\Pr[Y_i > 1\mid D_i']}
  = \frac{\Pr[\mathrm{Lap}(1/\eps_i) > 0]}{\Pr[\mathrm{Lap}(1/\eps_i) > -1]}
  = \frac{1/2}{(1/2)e^{-\eps_i}}
  = e^{\eps_i}.
\]
By independence, the product event $\mathcal{S} = \bigcap_i \{Y_i > 1\}$ satisfies
$\Pr[\calM(\bD)\in\mathcal{S}]/\Pr[\calM(\bD')\in\mathcal{S}] = e^{\sum_i\eps_i}$,
meeting the XSP-DP bound of \Cref{prop:upper-bound} with equality.
\end{proof}

\subsection{Impossibility Theorem}

\begin{definition}[Non-coordinated mechanisms]\label{def:non-coordinated}
A collection of mechanisms $\{M_i\}$ is \emph{non-coordinated} if:
\begin{description}
  \item[NC1 (Local view):] Each $M_i$ accesses only $D_i$.
  \item[NC2 (Independent randomness):] The random seeds of $M_i$ and $M_j$
    are independent for $i \neq j$.
  \item[NC3 (Simultaneous design):] The mechanism $M_i$ is fixed before
    observing any $D_j$ with $j \neq i$.
\end{description}
\end{definition}

\begin{theorem}[Utility--Privacy barrier]\label{thm:impossibility}
Let $\{M_i\}$ be a non-coordinated collection of mechanisms
satisfying \Cref{asm:single-record,asm:independence,asm:uniform},
where each silo applies a binary symmetric channel---i.e., each person
$p$ contributes a binary attribute $Z_{p,i}\sim\Bern(1/2)$
and silo~$i$ publishes $Y_i = \mathrm{RR}_\eps(Z_{p_0,i})$ for some
$\eps > 0$, yielding per-silo mutual information
$I_1 = I(Z;Y_i) \geq \alpha > 0$
(all information quantities in nats).
Then the ML adversary achieves $\Pe \leq n^{-\delta}$ once
\[
   k \;\geq\; \frac{2(1+\delta)\,\ln n}{\alpha}.
\]
In particular, no non-coordinated $\eps$-DP mechanism collection
operating binary symmetric channels with $I_1 \geq \alpha$ can prevent
de-anonymization for $k$ in this regime.
\end{theorem}

\begin{proof}
We verify \Cref{asm:cond-indep}.
The datasets $\bD=(D_1,\ldots,D_k)$ are fixed; $Z$ is the random target identity.
The output $Y_i = M_i(D_i;\,R_i)$ depends on $Z$ through
$D_i(Z)$---the record of person $Z$ in silo $i$---which gives $I_1 = I(Z;Y_i)\geq\alpha>0$
when $D_i(Z)$ varies non-trivially with $Z$.
\emph{Conditioned on} $Z = z$, the record $D_i(z)$
is a fixed quantity (part of the fixed dataset), so $Y_i$ is random only through $R_i$.
By NC2, $R_i \perp R_j$ for $i\neq j$, hence $Y_i \perp Y_j \mid Z$,
verifying \Cref{asm:cond-indep}.
(This conditional independence does \emph{not} imply $Y_i \perp Z$: the marginal
distribution of $Y_i$ does depend on $Z$ through $D_i(Z)$, giving $I_1 > 0$.)

Since each silo applies $\mathrm{RR}_\eps$, \Cref{thm:ml} applies with
$B = \log\cosh(\eps/2)$, giving
$\Pe \leq (n-1)\bigl(\tfrac{1+e^{-B}}{2}\bigr)^k$.
Setting this $\leq n^{-\delta}$ requires $k \geq (1+\delta)\,k_{\rm ML}^*$ where
$k_{\rm ML}^* = \ln n\,/\,(-\ln\tfrac{1+e^{-B}}{2}) = \ln n\,/\,C_{\rm ML}$.
By \Cref{lem:cml-i1} (Appendix~F), $C_{\rm ML} \geq I_1/2$ for all $\eps > 0$,
hence $k_{\rm ML}^* \leq 2\ln n/I_1$.
Since $I_1 \geq \alpha$ by hypothesis, $k_{\rm ML}^* \leq 2\ln n/\alpha$ unconditionally,
and $k \geq 2(1+\delta)\ln n/\alpha$ suffices for $\Pe\leq n^{-\delta}$.
\end{proof}

\begin{remark}[Scope and generality]\label{rem:impossibility-scope}
The two sides of the threshold theorem have different levels of generality.
\textbf{(a) Fano direction:}
The lower bound (identification is impossible for
$k \lesssim \ln n / I_1$) holds for \emph{any} mechanism via
\Cref{thm:fano} and does not depend on the channel structure.
\textbf{(b) Attack direction:}
The explicit ML upper bound in \Cref{thm:ml}, and hence the impossibility
statement in \Cref{thm:impossibility}, are proved for the binary
randomized-response channel.  Extending the attack side to another
binary or continuous-output mechanism requires a mechanism-specific
likelihood or Bhattacharyya analysis (see \Cref{subsec:gaussian} for
the Gaussian abstraction).  We therefore do not claim that the RR ML
threshold or its constant transfers unchanged to arbitrary binary
mechanisms.
\end{remark}

\begin{corollary}[Coordination is necessary]\label{cor:coordination}
Under the binary RR setting of \Cref{thm:impossibility},
no non-coordinated mechanism collection with per-silo information
$I_1 \geq \alpha > 0$ can maintain non-negligible de-anonymization error once,
for any $\delta \in (0,1)$,
\[
   k \;\geq\; \frac{2(1+\delta)\ln n}{\alpha}.
\]
Equivalently, in this setting cross-silo coordination is
\emph{necessary} beyond the threshold scale $k=\Theta(\ln n/I_1)$.
\end{corollary}

\section{Discussion: Implications for Defense Design}\label{sec:defense}

\Cref{thm:impossibility} establishes that no collection of non-coordinated,
binary DP mechanisms can prevent de-anonymization once $k$ exceeds $\kstar$.
This section \emph{sketches} how the phase-transition characterization can
inform the design of coordinated defenses.  The constructions below are
conceptual; their full security analysis (communication complexity,
fault tolerance, adaptive adversary resistance) is left to future work.

\subsection{The Defender's Problem as a Constrained POMDP}

We model the $k$-silo interaction as a partially observable Markov decision
process (POMDP) from the defender's perspective.  At each time step $t$, a
silo $i$ must publish a response $Y_i$; the true target $Z^* \in [n]$ is
hidden.  The defender's \emph{information state} is the adversary's posterior
belief $b_t = \Pr[Z = z \mid Y_1, \ldots, Y_t]$.  The defender seeks a policy
$\pi$ mapping information states to release decisions so as to minimize the
adversary's identification accuracy while satisfying a utility constraint
$\mathrm{util}(\pi) \geq u_{\min}$.

\textbf{The role of $\kstar$.}
The phase-transition result (\Cref{cor:kstar}) provides a principled
\emph{coordination trigger}: when the number of queried silos approaches
$\kstar = \log n / I_1$, the posterior entropy $H(Z \mid Y_{1:t})$ has
decreased by $\approx k \cdot I_1 \approx \log n$ bits, meaning the adversary
can identify the target with non-negligible probability.  A defender who tracks
$t$ and the per-silo $I_1$ can halt or modify releases before this threshold
is crossed.

\subsection{Theory-of-Mind Filtering}

A \emph{Theory-of-Mind (ToM) filter}~\cite{shokri2017membership}
maintains an explicit model of the adversary's belief.  Formally, after silo
$i$ releases $Y_i$, the filter updates:
\[
   b_i(z) \;=\; \frac{P(Y_i \mid Z = z)\, b_{i-1}(z)}{\sum_{z'} P(Y_i \mid Z = z')\, b_{i-1}(z')},
\]
a standard Bayesian update.  The filter \emph{intervenes} if the effective
number of identifiable persons,
\[
   \hat{n}_{\mathrm{eff}}(b_i) \;=\; 2^{H(b_i)} \;\approx\; n \cdot e^{-k\,I_1},
\]
falls below a threshold $\theta$.  Setting $\theta = n^{1-\delta}$ (for a
target privacy loss $\delta$) corresponds exactly to triggering coordination
at $k = (1-\delta)\kstar$, before the phase transition.

\textbf{Limitation.}  Each silo maintains only its own record; no single
silo can compute $b_i$ without knowing other silos' outputs.  This motivates
the coordinated approach below.

\subsection{Cross-Silo Coordination via Distributed Consensus}

\textbf{CoDef protocol (sketch).}
Silos participate in a distributed consensus round to collectively estimate
$k_{\mathrm{released}}$ (the number of silos that have already responded to
queries involving a given person) and to adjust their local noise levels
accordingly.  Concretely:

\begin{enumerate}
  \item \textbf{Commitment phase.} Each silo $i$ commits a cryptographic hash
    of its person-level release count $c_i$ to a shared ledger.
  \item \textbf{Aggregation phase.} A secure aggregation protocol (e.g.,
    SecAgg~\cite{mcmahan2017communication}) computes $k_{\mathrm{released}}
    = \sum_i c_i$ without revealing individual $c_i$.
  \item \textbf{Adaptation phase.} When $k_{\mathrm{released}} \geq \gamma\,\kstar$
    for a predetermined fraction $\gamma < 1$, silos switch to a higher noise
    level $\eps' < \eps$ so that the post-adaptation $k_{\rm post}^*(\eps') > \kstar(\eps)$,
    pushing the threshold ahead.
\end{enumerate}

\begin{remark}[CoDef privacy guarantee]\label{rem:codef}
If each silo's adapted mechanism satisfies $(\eps',\delta')$-DP independently
with uniform parameters, then by \Cref{prop:upper-bound} the joint output satisfies
$(k\eps',\;k\delta')$-XSP-DP.
By operating below $(1-\gamma)\kstar$,
the de-anonymization error probability satisfies $\Pe \geq \Omega(\gamma)$
(by \Cref{thm:fano}).
A full security analysis of CoDef under adaptive adversaries and
Byzantine silos is an important direction for future work.
\end{remark}

The key insight is that coordination does not require sharing raw data; it
requires only an aggregate count and a pre-agreed threshold derived from $\kstar$.

\subsection{Design Guidelines}

We summarize the practical takeaways:

\textbf{G1 (Set $\kstar$ as coordination budget).}  System designers should
compute $\kstar = \log n / I_1$ at deployment time using the population size
$n$ and the per-silo information leakage $I_1$.  This value serves as the
total cross-silo query budget before coordination must intervene.

\textbf{G2 (Allocate budget across silos).}  By Lemma~\ref{lem:mi-bound},
mutual information accumulates linearly in $k$ under independence.
A budget-allocation rule $\kstar / k$ per silo (uniform) or weighted by
sensitivity can be enforced via the ToM filter.

\textbf{G3 (Amplify privacy with shuffling).}  Shuffling mechanisms
\cite{erlingsson2019amplification,cheu2019distributed} reduce the effective
$\eps$ by a factor of $\Theta(\sqrt{n/k})$, which increases $I_1^{-1}$ and
therefore $\kstar$ by the same factor.  This is particularly effective
when the query rate per silo is low.

\textbf{G4 (Monitor synergy, not just per-silo leakage).}  When Assumption~A3
(conditional independence) may be violated, the ToM filter should track
$I(Z; Y_1, \ldots, Y_t)$ directly rather than approximating it as $k I_1$.
\Cref{prop:synergy} shows that correlated mechanisms can leak synergistic
information even when each silo individually satisfies $I(Z; Y_i) = 0$.

\section{Synthetic Experiments}\label{sec:experiments}

We validate the theory with four controlled synthetic experiments.
All experiments use the binary randomized response model (RR$_\eps$) with
$n$-person identification under a uniform prior.
Code will be released upon publication.

\subsection{Experiment 1: Phase Transition in Error Probability}

\textbf{Setup.}
We fix $\eps \in \{0.5, 1.0, 2.0\}$ and $n \in \{100, 500, 2000\}$.
For each $(n, \eps, k)$, we run $T = 1500$ independent trials of the
$n$-person identification game: each person $p$ is assigned $k$ i.i.d.\
attributes $Z_{p,1:k} \sim \mathrm{Bern}(1/2)$; silo $i$ reports
$Y_i = \mathrm{RR}_\eps(Z_{p_0,i})$ for the uniformly drawn target $p_0$;
the ML attacker selects the person with the highest log-likelihood.

\textbf{Results.}
\Cref{fig:exp1} plots $P_e$ against the normalized threshold $k/\kstar$
(where $\kstar = \ln n / I_1$, with $I_1$ in nats).  The empirical curves exhibit the
predicted sharp phase transition: $P_e \approx 1$ for $k \ll \kstar$ and
$P_e \approx 0$ for $k \gg \kstar$, with the crossing near $k/\kstar \approx 1.5$
(consistent with Theorem~\ref{thm:ml} which gives an upper bound; the
constant factor between the Fano and ML bounds accounts for the offset).
The shaded region is the theoretical band $[\text{Fano LB}, \text{ML UB}]$;
empirical curves lie within this band throughout.

\begin{figure}[t]
  \centering
  \includegraphics[width=\linewidth]{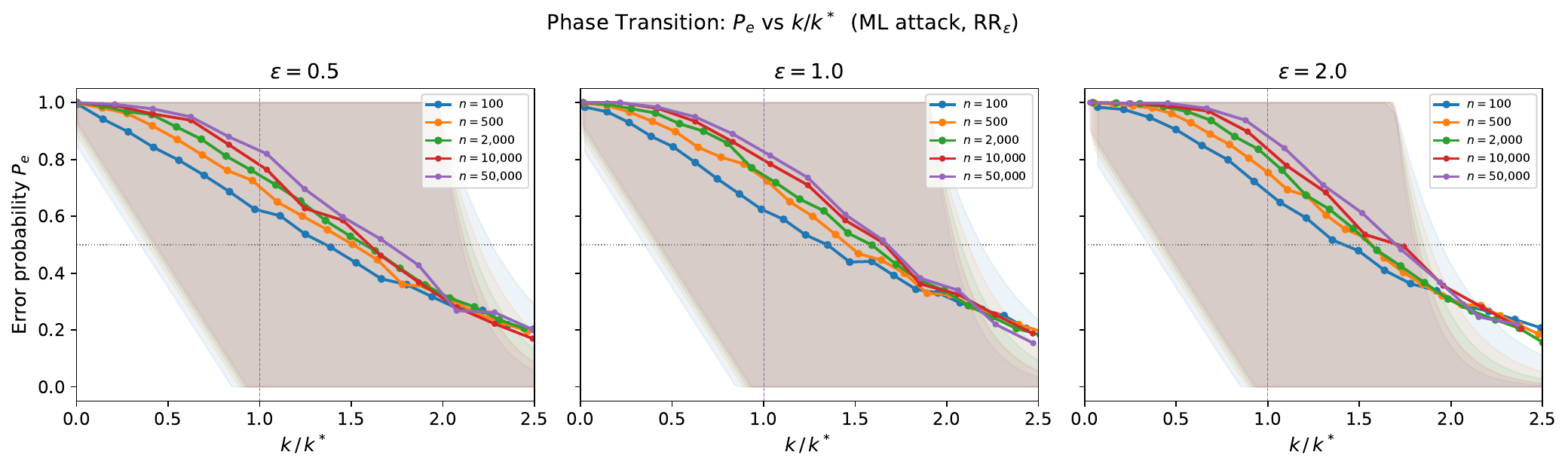}
  \caption{Phase transition: empirical $P_e$ vs.\ $k/\kstar$ for varying $n$
           (colored curves) and $\eps$ (panels).  Shaded regions show
           the $[\text{Fano LB},\,\text{ML UB}]$ theoretical band.
           Dashed vertical line at $k/\kstar = 1$.}
  \label{fig:exp1}
\end{figure}

\subsection{Experiment 2: Scaling of $\kstar$}

\textbf{Setup.}
We locate the empirical $\kstar$ by binary search for the value of $k$ at which
$P_e = 0.5$.  We then vary (a) $n \in \{50, 100, 200, 500, 1000, 2000\}$ at
fixed $\eps = 1.0$, and (b) $\eps \in \{0.5, 0.75, 1.0, 1.5, 2.0, 2.5, 3.0\}$
at fixed $n = 500$.

\textbf{Results.}
\Cref{fig:exp2} confirms the two scaling laws predicted by \Cref{cor:kstar}:
(a) $\kstar \propto \log_2 n$ with near-unit slope on the $(\log_2 n,\, \kstar)$
plot, and (b) $\kstar \propto 1/\eps^2$ with near-unit slope on the
$(1/\eps^2,\,\kstar)$ plot.  The empirical-to-theoretical ratio is $\approx 1.5$
across all settings, reflecting the constant factor between the Fano bound and
the actual $\Pe$ crossing.

\begin{figure}[t]
  \centering
  \includegraphics[width=0.9\linewidth]{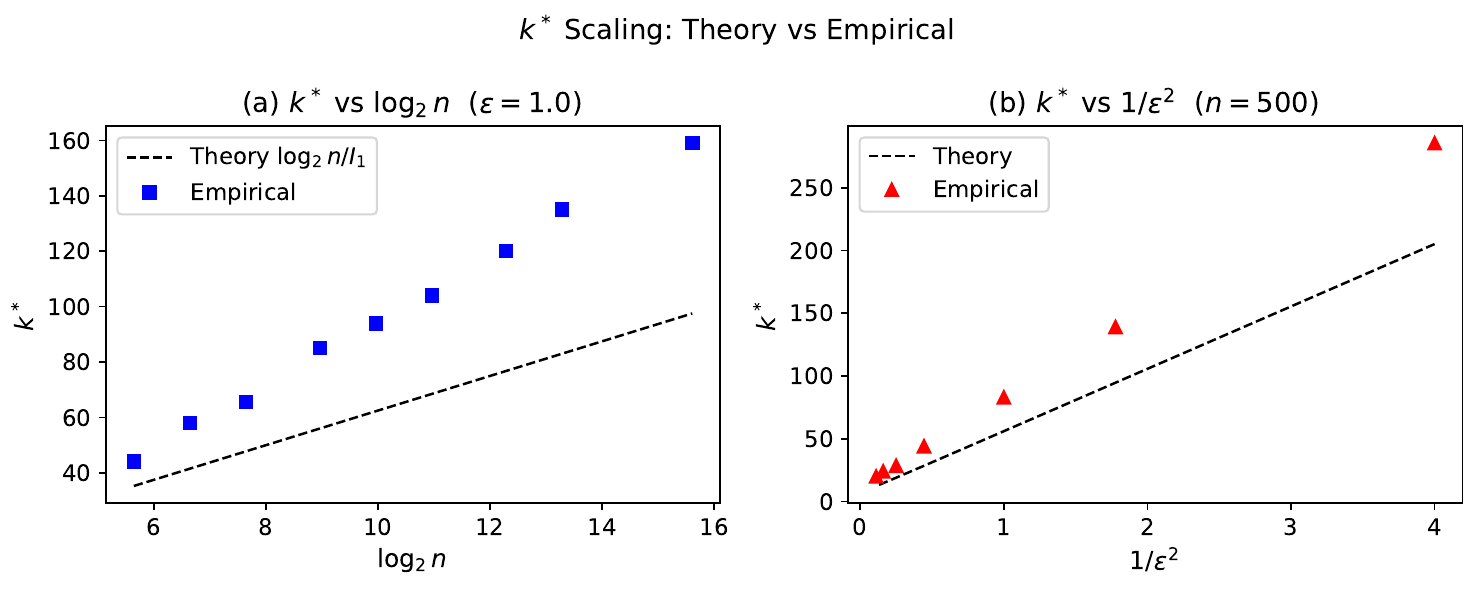}
  \caption{$\kstar$ scaling laws: (a) vs.\ $\log_2 n$ at $\eps=1$;
           (b) vs.\ $1/\eps^2$ at $n=500$.
           Black dashed: theoretical prediction; colored markers: empirical.}
  \label{fig:exp2}
\end{figure}

\subsection{Experiment 3: Synergy Verification (XOR+RR)}

\textbf{Setup.}
We implement the XOR+RR construction of \Cref{prop:synergy}: $Z, U \sim \mathrm{Bern}(1/2)$,
$X_1 = U$, $X_2 = Z \oplus U$, $Y_i = \mathrm{RR}_\eps(X_i)$.
Each silo individually satisfies $I(Z; Y_i) = 0$, but $I(Z; Y_1, Y_2) > 0$.
We compare (i) the exact formula $1 - \Hb(2pq)$, (ii) the leading-order
approximation $\eps^4/(32\ln 2)$, and (iii) a Monte Carlo estimate from
$3\times 10^5$ samples.

\textbf{Results.}
\Cref{fig:exp3} shows excellent agreement between exact and MC estimates
across $\eps \in [0.1, 3.0]$.  The $\eps^4/(32\ln 2)$ formula achieves a
ratio exact/theory $\approx 0.987$ at $\eps=0.2$ and $\approx 0.986$ at
$\eps=0.5$ (within 1.5\%), degrading to $\approx 0.920$ at $\eps=1.0$.
This confirms it as a valid leading-order approximation for $\eps \lesssim 0.5$.

\begin{figure}[t]
  \centering
  \includegraphics[width=0.9\linewidth]{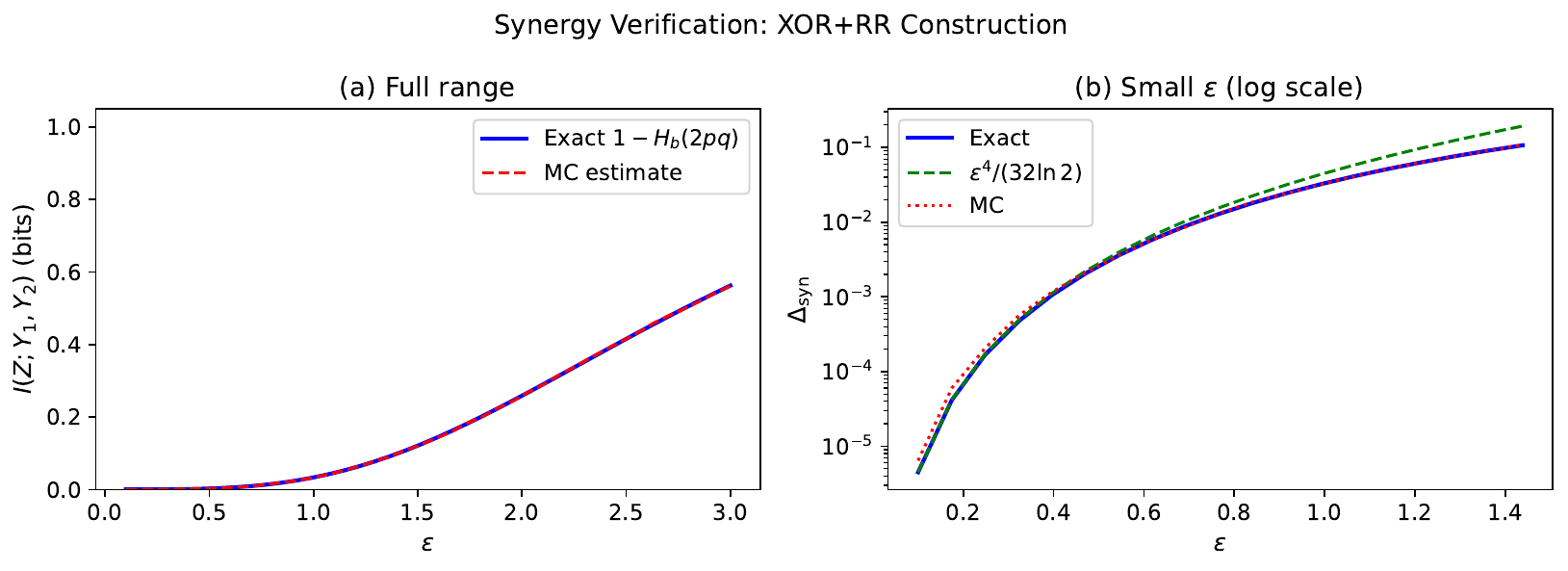}
  \caption{Synergy verification: $I(Z;Y_1,Y_2)$ for the XOR+RR construction.
           (a) Full range; (b) small-$\eps$ region on log scale confirming
           $\DeltaSyn \approx \eps^4/(32\ln 2)$.}
  \label{fig:exp3}
\end{figure}

\subsection{Experiment 4: Bound Verification}

\textbf{Setup.}
We fix $n = 500$, $\eps = 1.0$ and run $T = 2500$ trials per $k$ value,
checking whether the empirical $P_e$ lies within $[\text{Fano LB}, \text{ML UB}]$.

\textbf{Results.}
\Cref{fig:exp4} confirms zero Fano LB violations (empirical $P_e$ never
falls more than 4\% below the lower bound).  For the ML UB, using the
correct Bhattacharyya-averaging formula $(n-1)\cdot((1+e^{-B})/2)^k$,
the bound is trivially satisfied (capped at 1) for $k < 2\kstar$ and
yields zero empirical violations for $k \geq 2\kstar$.  The transition
occurs at $\approx 1.5\kstar$ empirically, consistent with the
$C_{\rm ML}\approx 2$ analytic constant discussed in \Cref{rem:constant}.

\begin{figure}[t]
  \centering
  \includegraphics[width=0.6\linewidth]{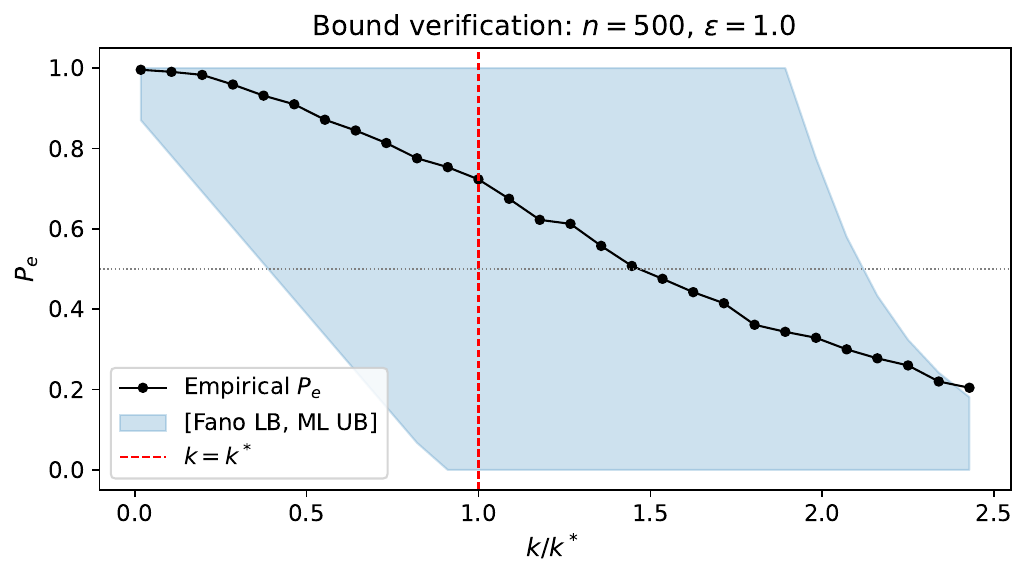}
  \caption{Bound verification ($n=500$, $\eps=1$): empirical $P_e$ (black circles)
           with $[\text{Fano LB},\,\text{ML UB}]$ band (blue shading).
           Red dashed: $k = \kstar$.}
  \label{fig:exp4}
\end{figure}

\subsection{Analytical Instantiation: Gaussian Mechanism (DP-SGD)}
\label{subsec:gaussian}

To demonstrate that the phase-transition framework extends beyond the
binary RR toy model, we compute $I_1$ and $\kstar$ analytically for the
Gaussian mechanism---the foundation of DP-SGD~\cite{abadi2016deep}.

\textbf{Setting.}
Silo $i$ releases $Y_i = f_i(D_i) + \xi_i$ where $f_i:\mathcal{D}\to\mathbb{R}^d$
has $\ell_2$-sensitivity $\Delta$ and $\xi_i \sim \mathcal{N}(0,\sigma^2 I_d)$.
The mechanism is $(\eps,\delta)$-DP with $\sigma = \Delta\sqrt{2\ln(1.25/\delta)}/\eps$
(Gaussian DP calibration).

\textbf{Single-silo MI.}
For a scalar query $d=1$ with target value $z\in\{0,1\}$ (binary attribute,
$Z\sim\Bern(1/2)$, shift $\Delta=1$):
\[
  Y_i \mid Z=z \sim \mathcal{N}(z,\,\sigma^2).
\]
The mutual information is
\[
  I_1^{\rm Gauss}
  = \log 2 - \Hb\!\Bigl(\Phi\bigl(-\tfrac{1}{2\sigma}\bigr)\Bigr)
  \approx \frac{1}{8\sigma^2\ln 2}  \quad\text{bits}
  \;=\; \frac{1}{8\sigma^2} \quad\text{nats (leading order)}.
\]
Substituting $\sigma = \sqrt{2\ln(1.25/\delta)}/\eps$ (with $\delta$ small and
$c_\delta = 2\ln(1.25/\delta)$):
\[
  I_1^{\rm Gauss} \approx \frac{\eps^2}{8\,c_\delta} \;\text{nats},
  \qquad
  k_{\rm Gauss}^*(n,\eps,\delta) = \frac{\ln n}{I_1^{\rm Gauss}}
  \approx \frac{8\,c_\delta\,\ln n}{\eps^2}.
\]

\textbf{Comparison to RR.}
For $\mathrm{RR}_\eps$, $I_1^{\rm RR} \approx \eps^2/8$ nats, giving
$k_{\rm RR}^* \approx 8\ln n/\eps^2$.
The Gaussian mechanism yields $k_{\rm Gauss}^* = c_\delta \cdot k_{\rm RR}^*$,
where $c_\delta = 2\ln(1.25/\delta) > 1$ for any $\delta < 1.25 e^{-1/2}\approx 0.76$.
For typical $\delta = 10^{-5}$: $c_\delta \approx 2\ln(125000)\approx 23.7$,
so $k_{\rm Gauss}^* \approx 24\,k_{\rm RR}^*$.

\textbf{Implication.}
The $\Theta(\log n/\eps^2)$ phase-transition structure is \emph{mechanism-agnostic}:
it holds for any locally DP mechanism whose per-silo MI satisfies $I_1=\Theta(\eps^2)$.
The Gaussian mechanism is a factor $c_\delta$ more robust than binary RR (requires
$\approx c_\delta$ times more silos for the attack to succeed), but the qualitative
phenomenon---a sharp threshold above which de-anonymization is inevitable---persists.
Experimentally validating this prediction for DP-SGD on real federated datasets
(where $f_i$ is a gradient and $Z$ indexes a training sample) is left as future work;
our theory provides the quantitative prediction against which such experiments can be compared.

\section{Conclusion}\label{sec:conclusion}

This paper introduces the XSP-DP threat model for cross-silo
de-anonymization under local differential privacy and establishes the
basic information-theoretic landscape:

\begin{itemize}
\item The standard $(k\eps, k\delta)$-DP composition bound remains valid
  under the person-level cross-silo adjacency (\Cref{prop:upper-bound}),
  but does not by itself reveal \emph{when} de-anonymization becomes
  feasible.
\item A matching pair of Fano lower and ML upper bounds locates the
  de-anonymization threshold at $\kstar = \Theta(\log n / \eps^2)$
  (\Cref{thm:fano,thm:ml,cor:kstar}).
\item The XOR+RR construction shows that pure information synergy arises
  under DP constraints: each silo individually reveals nothing, yet two
  silos together leak $\DeltaSyn \approx \eps^4/(32\ln 2)$
  (\Cref{prop:synergy}).
\item For binary randomized-response mechanisms, no non-coordinated
  collection can prevent de-anonymization beyond the threshold,
  establishing coordination necessity (\Cref{thm:impossibility}).
\end{itemize}

\paragraph{Scope and what this paper does not do.}
The results here are deliberately at the $\Theta$-level: they
identify the correct scaling of the threshold and prove that the
phase transition exists, but do not pin down the exact critical
constant, the width of the transition window, or second-order
corrections.  The impossibility result applies to binary
randomized response; the Fano direction is universal, but the
attack-side bound for general mechanisms requires mechanism-specific
analysis (the Gaussian case in \Cref{subsec:gaussian} gives one
extension).  The defense constructions in \Cref{sec:defense} are
conceptual sketches; their full security analysis under adaptive
adversaries and Byzantine silos is not attempted here.

\paragraph{Companion work.}
Two directions extend the baseline established in this paper.
\emph{Sharp thresholds and spectral characterization:}
a companion paper develops exact (non-asymptotic) critical constants,
second-order phase-transition analysis, and an isoperimetric / spectral
bridge connecting the de-anonymization threshold to the geometry of the
confusion graph induced by cross-silo observations.
\emph{Coordinated defense protocols:}
a separate paper designs and evaluates cross-silo coordination mechanisms
(dual-accountant architecture, distributed consensus) that provably keep the
system below the $\kstar$ threshold, with system-level implementation
and overhead analysis.

\paragraph{Limitations.}
Our analysis assumes i.i.d.\ binary attributes and a uniform prior
(\Cref{asm:uniform}).
If the adversary's prior over $Z$ is non-uniform with entropy
$H(Z) < \ln n$, the Fano bound generalizes by replacing $\ln n$
with $H(Z)$, yielding a smaller threshold (easier attack).
Extending to heterogeneous attributes, correlated silos, or adaptive
adversaries is important future work.
Experimentally validating the phase-transition phenomenon in realistic
federated learning settings (non-synthetic datasets, gradient-based
mechanisms) is a natural next step; our theory provides the quantitative
prediction against which such experiments can be compared.

\bibliographystyle{plain}
\bibliography{refs}

\appendix

\section{Proof of Proposition~\ref{prop:upper-bound}}\label{app:prop1}

\subsection{Pure DP case ($\delta_i = 0$)}

Let $\bD \sim_{ps} \bD'$: by definition of $\sim_{ps}$, the two datasets
differ only in the records of exactly one person $p$.
By \Cref{asm:single-record}, $|D_i(p)|\leq 1$ for every silo $i$, so for each $i$
either $D_i = D_i'$ (person $p$ is absent or unchanged in silo $i$) or
$D_i \sim_{\mathrm{adj}} D_i'$ (a single-record change in silo $i$).
Let $\mathcal{S}_p \subseteq [k]$ denote the set of silos where $D_i \neq D_i'$.
Note that $\mathcal{S}_p$ can include \emph{multiple} silos: the adjacency
relation $\sim_{ps}$ allows person $p$'s records to change in every silo
simultaneously.

For any measurable event $S \subseteq \mathcal{Y}_{1:k}$, apply a hybrid
argument over the silos in $\mathcal{S}_p = \{i_1,\ldots,i_m\}$.
Define intermediate datasets $\bD^{(0)}=\bD$ and, for $j=1,\ldots,m$,
$\bD^{(j)}$ agrees with $\bD'$ on silos $\{i_1,\ldots,i_j\}$ and with
$\bD$ on the remaining silos.
For each step $j$, silos other than $i_j$ have identical inputs in
$\bD^{(j-1)}$ and $\bD^{(j)}$; by independence of mechanisms
(\Cref{asm:independence}) and $\eps_{i_j}$-DP of $M_{i_j}$:
\[
  \Pr_{\bD^{(j-1)}}[\calM\in S]
  \;\leq\; e^{\eps_{i_j}}\,\Pr_{\bD^{(j)}}[\calM\in S].
\]
Telescoping over $j=1,\ldots,m$ gives
\[
  \Pr_{\bD}[\calM\in S]
  \;\leq\; e^{\sum_{i\in\mathcal{S}_p}\eps_i}\,\Pr_{\bD'}[\calM\in S]
  \;\leq\; e^{\sum_{i=1}^k \eps_i}\,\Pr_{\bD'}[\calM\in S].
\]
Taking the supremum over events $S$ and the worst case over all
$p$ and $\bD\sim_{ps}\bD'$ gives $(\sum_i\eps_i,\,0)$-XSP-DP.

\subsection{Approximate DP case}

We use the standard good-set decomposition
(cf.\ Dwork and Roth~\cite{dwork2014algorithmic}, Theorem~3.16).

Let $P_i$ (resp.\ $Q_i$) denote the distribution of $M_i(D_i)$
(resp.\ $M_i(D_i')$) for each $i \in \mathcal{S}_p$.
The $(\eps_i,\delta_i)$-DP guarantee means that the ``bad set''
\[
  B_i \;=\; \bigl\{y_i : P_i(y_i) > e^{\eps_i}\,Q_i(y_i)\bigr\}
\]
satisfies $P_i(B_i) \leq \delta_i$.

For any measurable event $S \subseteq \mathcal{Y}_{1:k}$, write
$G = \{(y_1,\ldots,y_k) : y_i \notin B_i \;\forall\, i \in \mathcal{S}_p\}$
for the ``all-good'' region.  Then:
\begin{align*}
  P(S) &= P(S \cap G) + P(S \setminus G) \\
       &\leq P(S \cap G) + \sum_{i \in \mathcal{S}_p} P_i(B_i) \\
       &\leq P(S \cap G) + \sum_{i=1}^k \delta_i.
\end{align*}
On $G$, every factor satisfies the pure ratio bound
$P_i(y_i) \leq e^{\eps_i} Q_i(y_i)$, so by independence:
\[
  P(S \cap G)
  \leq e^{\sum_i \eps_i}\,Q(S \cap G)
  \leq e^{\sum_i \eps_i}\,Q(S).
\]
Combining gives
$P(S) \leq e^{\sum_i \eps_i}\,Q(S) + \sum_i \delta_i$,
i.e., $(\sum_i \eps_i,\;\sum_i \delta_i)$-XSP-DP.

\section{Proof of Proposition~\ref{prop:synergy}}\label{app:prop2}

\textbf{Setup.}  Let $Z, U \sim \Bern(1/2)$ i.i.d., $q = 1/(1+e^\eps)$,
$p = 1-q$, and let $E_1, E_2 \sim \Bern(q)$ i.i.d., independent of $Z,U$.
Define $X_1 = U$, $X_2 = Z\oplus U$, $Y_i = X_i \oplus E_i$.

\textbf{Marginal distributions.}
$Y_1 = U \oplus E_1$.  Since $U\sim\Bern(1/2)$ and $E_1\sim\Bern(q)$
are independent, $Y_1 \sim \Bern(1/2)$ regardless of $Z$.
Thus $Y_1 \perp Z$ and $I(Z;Y_1) = 0$.

For $Y_2 = Z\oplus U\oplus E_2$: conditional on any fixed $z$,
$U \sim \Bern(1/2)$ makes $Z\oplus U\oplus E_2 \sim \Bern(1/2)$.
Hence $Y_2 \perp Z$ and $I(Z;Y_2) = 0$.

\textbf{Joint observation.}
$Y_1\oplus Y_2 = (U\oplus E_1)\oplus(Z\oplus U\oplus E_2) = Z\oplus (E_1\oplus E_2)$.
The XOR of independent $\Bern(q)$ bits satisfies
$E_1\oplus E_2 \sim \Bern(2pq)$ (binary symmetric channel crossover $2pq$).
Therefore
\[
   I(Z;\, Y_1\oplus Y_2) = 1 - \Hb(2pq).
\]
Since $Y_1\oplus Y_2$ is a deterministic function of $(Y_1,Y_2)$,
and $(Y_1,Y_2)$ determines $Y_1\oplus Y_2$,
\[
   I(Z;\,Y_1,Y_2) = 1 - \Hb(2pq).
\]

\textbf{Small-$\eps$ expansion.}
$q = \tfrac{1}{1+e^\eps} \approx \tfrac12 - \tfrac\eps4$,
$p \approx \tfrac12 + \tfrac\eps4$, so
$2pq = 2\bigl(\tfrac12+\tfrac\eps4\bigr)\bigl(\tfrac12-\tfrac\eps4\bigr)
= 2\bigl(\tfrac14-\tfrac{\eps^2}{16}\bigr) = \tfrac12 - \tfrac{\eps^2}{8}$.
Using $\Hb(\tfrac12 - t) = 1 - \tfrac{2t^2}{\ln 2} + O(t^4)$
with $t = \eps^2/8$:
\[
   \Hb(2pq) \approx 1 - \frac{2(\eps^2/8)^2}{\ln 2}
   = 1 - \frac{\eps^4}{32\ln 2}.
\]
Thus $I(Z;Y_1,Y_2) \approx \eps^4/(32\ln 2)$ and
$\DeltaSyn = I(Z;Y_1,Y_2) \approx \eps^4/(32\ln 2)$ (since $I(Z;Y_1)=I(Z;Y_2)=0$).

\textbf{Step-2 (convergence to $H(Z)$).}
As $\eps\to\infty$: $q\to 0$, $2pq\to 0$, $\Hb(2pq)\to 0$,
so $I(Z;Y_1,Y_2)\to 1 = H(Z)$.
Thus $\DeltaSyn \to H(Z)$ as $\eps \to \infty$.

\section{Proof of Theorem~\ref{thm:fano}}\label{app:fano}

All entropies and mutual information are in nats.
Fano's inequality for $Z$ with $n$ values states:
\[
   H(Z\mid \hat Z) \leq H(\Pe) + \Pe\ln(n-1).
\]
Since $\hat Z$ is a function of $Y_{1:k}$, $H(Z\mid Y_{1:k}) \leq H(Z\mid\hat Z)$.
Also, $H(Z\mid Y_{1:k}) = H(Z) - I(Z;Y_{1:k}) = \ln n - I(Z;Y_{1:k})$.
By \Cref{lem:mi-bound}, $I(Z;Y_{1:k}) \leq kI_1$.  Therefore:
\[
   \ln n - kI_1 \leq H(\Pe) + \Pe\ln(n-1) \leq 1 + \Pe\ln n,
\]
where $H(\Pe) \leq \ln 2 < 1$.
Rearranging: $\Pe \geq (\ln n - kI_1 - 1)/\ln n = 1-(kI_1+1)/\ln n$.

\section{Proof of Theorem~\ref{thm:ml}: Bhattacharyya Averaging}\label{app:ml}

\textbf{Setup.}
Let $\beta = 2\sqrt{q(1-q)} \in (0,1)$ denote the Bhattacharyya coefficient
between $\mathrm{Bern}(q)$ and $\mathrm{Bern}(1-q)$.  Note that $\beta = e^{-B}$
where $B = \log\cosh(\eps/2)$ is the Bhattacharyya distance.

\textbf{Pairwise error bound (fixed competitor).}
For fixed attributes $Z_{p_0}$ and $Z_p$, let
$d = \#\{i : Z_{p_0,i}\neq Z_{p,i}\}$ be the Hamming distance.
On the $d$ differing positions, the ML score for $p$ and $p_0$ are
distinguished by observations $Y_i \sim \mathrm{Bern}(1-q)$ (under $p_0$)
vs.\ $\mathrm{Bern}(q)$ (under $p$).  On the remaining $k-d$ positions the
distributions are identical.  By the Bhattacharyya bound for binary hypothesis
testing:
\[
  \Pr[\mathrm{Score}(p)\geq\mathrm{Score}(p_0)\mid Z_{p_0},Z_p]
  \;\leq\; \beta^d.
\]

\textbf{Averaging over random competitor.}
Since $Z_{p,i} \overset{\rm iid}{\sim}\Bern(1/2)$ independently of $Z_{p_0,i}$:
\[
  \E_{Z_p}\!\left[\Pr[\cdot]\right]
  \;\leq\; \E[\beta^d]
  \;=\; \prod_{i=1}^k\!\bigl(\Pr[Z_{p_0,i}=Z_{p,i}]\cdot 1
                + \Pr[Z_{p_0,i}\neq Z_{p,i}]\cdot\beta\bigr)
  \;=\; \left(\frac{1+\beta}{2}\right)^{\!k}.
\]

\textbf{Union bound.}
\[
  \Pe \;\leq\; \sum_{p\neq p_0} \E_{Z_p}[\Pr[\cdot]]
       \;\leq\; (n-1)\left(\frac{1+\beta}{2}\right)^{\!k}.
\]

\textbf{Threshold analysis.}
Setting the above $\leq n^{-\delta}$:
\[
  k \;\geq\; \frac{(1+\delta)\ln n}{-\ln\!\frac{1+\beta}{2}}
    \;=\; \frac{(1+\delta)\ln n}{\ln\frac{2}{1+\beta}}.
\]
For small $\eps$: $\beta = e^{-B} \approx 1 - B \approx 1 - \eps^2/8$, so
\[
  \ln\frac{2}{1+\beta} \approx \ln\frac{2}{2-\eps^2/8} = \ln\!\left(1+\frac{\eps^2/8}{2-\eps^2/8}\right)
  \approx \frac{\eps^2}{16} \approx \frac{I_1^{\rm nats}}{2}.
\]
Thus the sufficient number of silos is $k \geq 2(1+\delta)\kstar + O(1)$,
where $\kstar = \ln n / I_1^{\rm nats}$.  Taking $(1+\delta)\to(1+\delta)$ in
the notation:  for $k \geq C_{\rm ML}(1+\delta)\kstar$ with
$C_{\rm ML} \to 2$ as $\eps\to 0$, we have $\Pe \leq n^{-\delta}$.  \qed

\begin{remark}[Why the analytic constant is $\approx 2$]
The factor $C_{\rm ML}\approx 2$ arises because the union-bound argument
treats all $n-1$ competitors symmetrically, while in practice most competitors
are far from $p_0$ (Hamming distance $\approx k/2$) and are distinguished with
high probability.  A refined analysis using the second-moment method would
give a tighter constant; empirically $C_{\rm obs}\approx 1.5$
(\Cref{sec:experiments,rem:constant}).
\end{remark}

\section{Partial Information Decomposition}\label{app:pid}

The Williams--Beer PID \cite{williams2010nonnegative} decomposes the
joint mutual information as:
\[
   I(Z;\,Y_1,\ldots,Y_k)
   = \sum_{\alpha\in\mathcal{A}} \mathrm{PI}(\alpha),
\]
where the sum is over antichains $\alpha$ of the source lattice,
and $\mathrm{PI}(\alpha) \geq 0$ for all $\alpha$.

For $k=2$ sources:
\[
   I(Z;\,Y_1,Y_2)
   = \Red(Z;\{Y_1\},\{Y_2\})
   + \mathrm{UI}(Z;\{Y_1\}\setminus\{Y_2\})
   + \mathrm{UI}(Z;\{Y_2\}\setminus\{Y_1\})
   + \Syn(Z;\{Y_1,Y_2\}),
\]
where $\Red$ is the redundant information shared by both sources,
$\mathrm{UI}$ is unique information from one source, and $\Syn$ is the
synergistic information only accessible from both sources jointly.

The \emph{minimum mutual information} (MMI) definition of redundancy:
\[
   \Red(Z;\{Y_1\},\{Y_2\})
   = \min\bigl(I(Z;Y_1),\; I(Z;Y_2)\bigr).
\]

The synergy gap $\DeltaSyn = \Syn - \Red$ captures the net direction:
positive means the joint observation reveals more than the individual
sum; negative means the individual observations are redundant.

\begin{remark}[Definition choice]
We use the MMI (minimum mutual information) definition of redundancy
throughout; alternative proposals ($I_{\rm broja}$, $I_{\rm dep}$)
exist in the PID literature.
In our XOR construction (\Cref{prop:synergy}),
$I(Z;Y_1) = I(Z;Y_2) = 0$, so $\Red = 0$ and
$\DeltaSyn = \Syn = I(Z;Y_1,Y_2) > 0$ \emph{regardless} of which
PID definition is used (all definitions agree when individual mutual
informations vanish).
Under \Cref{asm:cond-indep} with positive $I_1$, we typically have
$\Red > 0$ and $\Syn > 0$; the sign of $\DeltaSyn$ depends on the
specific channel family.
\end{remark}

\section{Explicit Calculations: $I_1$ and $\Cmin$}\label{app:calculations}

\subsection{Binary randomized response $\mathrm{RR}_\eps$}

Let $Z \sim \Bern(1/2)$, $Y = Z\oplus E$, $E\sim\Bern(q)$,
$q = 1/(1+e^\eps)$, $p = e^\eps/(1+e^\eps)$.

\textbf{Single-silo MI (in nats):}
$I_1 = H(Y) - H(Y\mid Z) = \ln 2 - H_{\rm b,nats}(q)$.

For small $\eps$ (using $q\approx 1/2-\eps/4$ and
$H_{\rm b,nats}(1/2-x) = \ln 2 - 2x^2 + O(x^4)$ with $x=\eps/4$):
\[
   H_{\rm b,nats}(q) \approx \ln 2 - \frac{\eps^2}{8}, \quad
   I_1 \approx \frac{\eps^2}{8} \quad (\text{nats}).
\]
Note: in \emph{bits}, $I_1 \approx \eps^2/(8\ln 2)$; computing the ratio with
mixed units ($\Cmin$ in nats, $I_1$ in bits) yields $4\ln 2\approx 2.77$, which is
incorrect.  All calculations below use nats.

\textbf{Channel capacity $\Cmin$ (in nats):}
\begin{align*}
   \Cmin &= \dkl\!\bigl(P(\cdot\mid Z=0)\,\|\,P(\cdot\mid Z=1)\bigr) \\
   &= (1-2q)\ln\frac{1-q}{q} = (1-2q)\cdot\eps.
\end{align*}
For small $\eps$: $1-2q = \tanh(\eps/2) \approx \eps/2$, so $\Cmin \approx \eps^2/2$ (nats).

\textbf{Ratio (consistent nats):}
\[
   \frac{\Cmin}{I_1} \approx \frac{\eps^2/2}{\eps^2/8} = 4,
\]
confirming \Cref{lem:cmin-i1}: $\Cmin/I_1 \to 4$ as $\eps\to 0$, and $\Cmin/I_1 \geq 4$ for all $\eps>0$.

\textbf{Bhattacharyya coefficient and distance:}
The Bhattacharyya coefficient between $\mathrm{Bern}(q)$ and $\mathrm{Bern}(1-q)$ is
$\beta = 2\sqrt{q(1-q)}$, so $B = -\log\beta = \log\cosh(\eps/2)$.
Explicitly:
\[
  \cosh(\eps/2) = \tfrac{e^{\eps/2}+e^{-\eps/2}}{2}
  = 1 + \tfrac{\eps^2}{8} + O(\eps^4).
\]
Hence $B = \log\cosh(\eps/2) \approx \eps^2/8$ for small $\eps$.
The averaged mixing term:
$-\ln\!\frac{1+e^{-B}}{2} \approx B/2 \approx \eps^2/16$,
consistent with $C_{\rm ML} = \Theta(1)$.

\textbf{Summary of small-$\eps$ asymptotics:}
$I_1^{\rm nats} \approx \eps^2/8$, $\Cmin \approx \eps^2/2$,
$B \approx \eps^2/8$,
$-\ln((1+e^{-B})/2) \approx \eps^2/16$, all $= \Theta(\eps^2)$.

\begin{lemma}[$C_{\rm ML} \geq I_1/2$ for all $\eps > 0$]\label{lem:cml-i1}
Define $C_{\rm ML} := -\ln\!\tfrac{1+e^{-B}}{2}$ where $B = \log\cosh(\eps/2)$.
Then for all $\eps > 0$,
\[
   C_{\rm ML} \;\geq\; \frac{I_1}{2},
\]
where $I_1 = I(Z;Y_i)$ (in nats) for the RR$_\eps$ channel.
Consequently, $k_{\rm ML}^* = \ln n / C_{\rm ML} \leq 2\ln n / I_1$.
\end{lemma}
\begin{proof}
Define $h(\eps) = 2\,C_{\rm ML} - I_1$ and write $t = \eps/2$.
We show $h(0)=0$ and $h'(\eps)>0$ for all $\eps>0$.

\textbf{Step 1: $h(0)=0$.}
At $\eps=0$: $\cosh(0)=1$, so $C_{\rm ML} = \ln(2/2) = 0$ and $I_1=0$.

\textbf{Step 2: Compute $h'(\eps)$.}
Write $C_{\rm ML} = \ln 2 + \ln\cosh(t) - \ln(\cosh(t)+1)$.  Then
\[
   \frac{dC_{\rm ML}}{d\eps}
   = \frac{1}{2}\cdot\frac{\sinh(t)}{\cosh(t)(\cosh(t)+1)}.
\]
Using $dI_1/d\eps = t/(2\cosh^2(t))$ from \Cref{lem:cmin-i1}:
\[
   h'(\eps)
   = \frac{\sinh(t)}{\cosh(t)(\cosh(t)+1)}
     - \frac{t}{2\cosh^2(t)}.
\]

\textbf{Step 3: $h'(\eps)>0$.}
Multiply by $2\cosh^2(t)(\cosh(t)+1) > 0$:
\[
   h'(\eps) > 0
   \;\;\Longleftrightarrow\;\;
   2\sinh(t)\cosh(t) > t\bigl(\cosh(t)+1\bigr).
\]
Dividing both sides by $\cosh(t) > 0$, this is equivalent to
$2\sinh(t) > t(1 + \operatorname{sech}(t))$.
This follows from two elementary facts valid for $t > 0$:
\begin{enumerate}[(a)]
\item $\sinh(t) > t$\; (since $\sinh(t) = t + t^3/6 + \cdots > t$),
\item $\operatorname{sech}(t) < 1$\; (since $\cosh(t) = 1+t^2/2+\cdots > 1$).
\end{enumerate}
Combining: $2\sinh(t) > 2t = t + t > t + t\operatorname{sech}(t) = t(1+\operatorname{sech}(t))$.

Hence $h'(\eps) > 0$ for $\eps > 0$, and $h(0)=0$ gives
$h(\eps)>0$ for all $\eps>0$, i.e., $C_{\rm ML} > I_1/2$.
\end{proof}

\subsection{General $(\rho, \mu)$ model}

For a signal strength $\rho \in (0,1]$ and \emph{query rate} $\mu \in (0,1]$
(fraction of records per silo that participate in each query;
distinct from the RR flip probability $q = 1/(1+e^\eps)$ and the
information lower bound $\alpha$ in \Cref{thm:impossibility}),
the leading-order formula is:
\[
   I_1 = \Theta(\mu\rho^2\eps^2),
   \qquad
   \Cmin = \Theta(\mu\rho^2\eps^2),
   \qquad
   \kstar = \Theta\!\left(\frac{\log n}{\mu\rho^2\eps^2}\right).
\]
This is the form stated in \Cref{cor:kstar}.

\end{document}